\documentclass[lettersize,journal]{IEEEtran}
\usepackage{amsmath,amsfonts,amssymb,amsthm, amssymb} 
\usepackage{algorithmic}
\usepackage{algorithm}
\usepackage{array}
\usepackage[caption=false,font=normalsize,labelfont=sf,textfont=sf]{subfig}
\usepackage{textcomp}
\usepackage{stfloats}
\usepackage{url}
\usepackage{verbatim}
\usepackage{graphicx}
\usepackage{cite}
\usepackage{bm}
\usepackage{lipsum}
\usepackage{booktabs}
\usepackage{hyperref}
\hypersetup
{
	colorlinks=true,
	linkcolor=blue,
	filecolor=blue,
	urlcolor=blue,
	citecolor=cyan,
}
\usepackage{mathtools} 
\theoremstyle{definition}

\newtheorem{lemma}{Lemma}

\newtheorem{proposition}{Proposition}
\newtheorem{remark}{Remark}

\DeclareMathOperator*{\argmax}{arg\,max}

\newcommand{\Psiop}{\Psi} 

\begin{document}

\title{Waveform Design Based on Mutual Information Upper Bound For Joint Detection and Estimation}

\author{Ruofeng Yu, 
	Caiguang Zhang, 
	Chenyang Luo, 
	Mengdi Bai, 
	Shangqu Yan,  \\
	Wei Yang, and
	Yaowen Fu
	\thanks{This work was supported by the National Natural Science Foundation of China (Grant 61871384), the Science Fund for Distinguished Young Scholars of Hunan Province (Grant 2024JJ2066), and the Science and Technology Innovation Program of Hunan Province (Grant 2022RC1092). {{\itshape (Corresponding author: Yaowen Fu)}}}
	\thanks{All authors are with the College of Electronic Science and Technology, National University of Defense Technology, Changsha 410073, China (e-mail: {yuruofeng17@outlook.com; luochenyang19@nudt.edu.cn; baimengdi@nudt.edu.cn; shangqu\_yan@163.com; yw850716@sina.com; fuyaowen@nudt.edu.cn}). }
	\thanks{Caiguang Zhang is with the Shanghai Radio Equipment Research Institute, Shanghai 201100, China (e-mail: zhangiguang@163.com).}
}

\markboth{~Vol.~14, No.~8, August~2021}%
{Yu \MakeLowercase{\textit{et al.}}: Waveform Design Based on Mutual Information Upper Bound For Joint Detection and Estimation}


\maketitle

\begin{abstract}
Adaptive radar waveform design grounded in information-theoretic principles is critical for advancing cognitive radar performance in complex environments. This paper investigates the optimization of phase-coded waveforms under constant modulus constraints to jointly enhance target detection and parameter estimation. We introduce a unified design framework based on maximizing a Mutual Information Upper Bound (MIUB), which inherently reconciles the trade-off between detection sensitivity and estimation precision without relying on ad hoc weighting schemes. To model realistic, potentially non-Gaussian statistics of target returns and clutter, we adopt Gaussian Mixture Distributions (GMDs), enabling analytically tractable approximations of the MIUB's constituent Kullback-Leibler divergence and mutual information terms. To address the resulting non-convex problem, we propose the Phase-Coded Dream Optimization Algorithm (PC-DOA), a tailored metaheuristic that leverages hybrid initialization and adaptive exploration-exploitation mechanisms specifically designed for phase-variable optimization. Numerical simulations demonstrate the effectiveness of the proposed method in achieving modestly better detection-estimation trade-off.
\end{abstract}


\section{Introduction}\label{sec1}
\IEEEPARstart{M}{odern} radar systems constitute indispensable sensing platforms across diverse mission-critical domains, including strategic defense, autonomous navigation, remote sensing, and environmental monitoring \cite{skolnik_radar_2008}. The efficacy of these systems, particularly concerning fundamental performance metrics such as target detection probability, parameter estimation accuracy, and resolution capabilities, is inextricably linked to the characteristics of the transmitted waveforms \cite{levanon_radar_2004}. Recent advancements in cognitive radar paradigms have catalyzed a shift from employing static, predefined waveform libraries towards dynamically adaptive waveform synthesis, tailored in response to the operational environment and sensing objectives \cite{haykin_cognitive_2006}. This evolution is paramount for contemporary multifunction radar systems required to perform simultaneous detection, tracking, imaging, and classification tasks, often within increasingly congested and contested spectral environments \cite{guerci_cognitive_2014}.

Target detection, as a cornerstone function of radar systems, has long been the primary focus of waveform optimization research. The Neyman-Pearson (NP) criterion provides the theoretical foundation for optimal detector design, aiming to maximize detection probability under a fixed false alarm constraint. However, its practical implementation is hindered by the intractability of the likelihood ratio test, which often involves solving non-analytic integrals in realistic radar environments. As a result, surrogate metrics that maintain monotonic relationships with detection probability have been explored. Classical studies demonstrate that maximizing signal-to-clutter-plus-noise ratio (SCNR)~\cite{xu_probabilistically_2022,xu_radar_2025} can yield computationally efficient solutions for certain target and clutter models in additive white Gaussian noise (AWGN) channels. However, SCNR-based optimization does not necessarily guarantee the optimal detection performance. Alternative metrics such as deflection coefficients~\cite{zhu_information_2017} have been introduced to enhance detection robustness for weak target, while relative entropy metrics~\cite{tang_relative_2015,xu_relative_2025} have been utilized to optimize the separability between target-present and target-absent hypotheses. These techniques are subsequently implemented in advanced radar architectures, such as frequency diverse arrays (FDA)~\cite{gui_2017_coherent} and multiple-input multiple-output (MIMO) radar systems~\cite{tang_efficient_2018}. Additionally, a variety of sophisticated modulation schemes, including linear frequency modulation (LFM), polyphase-coded FM (PCFM)~\cite{blunt_PCFM_2014}, and orthogonal frequency division multiplexing (OFDM)~\cite{sen_2010_multiobjective}, have been explored to enhance transmitter adaptability across different operational scenarios.
Beyond detection, precise parameter estimation (e.g., range, Doppler, angle, scattering coefficients) is another critical radar function driving waveform design. Early work, such as DeLong's analysis \cite{delong_optimum_1970}, highlighted the influence of LFM waveform parameters on range-Doppler coupling and estimation accuracy via the Fisher Information Matrix (FIM). Contemporary estimation-centric waveform optimization primarily follows two main directions: (1) Minimizing estimation error bounds, such as the Cramér-Rao Lower Bound (CRLB) \cite{van_2023_transmit} or Mean Squared Error (MSE) \cite{tang_constrained_2021}, typically by optimizing scalar functions (e.g., trace or determinant) of the FIM; and (2) Maximizing the Mutual Information (MI) between the target parameters and the received signal \cite{bell_information_1993, iidriss_waveform_2021}, thereby minimizing the posterior uncertainty about the parameters given the observations. 

However, in many practical radar scenarios, detection and estimation are intrinsically coupled tasks. Target detection often necessitates prior or concurrent estimation of unknown parameters, while reliable detection outcomes can subsequently refine parameter estimates. This inherent interdependency motivates the pursuit of joint waveform optimization strategies that explicitly address the trade-offs between these objectives. Formulating this as a multi-objective optimization problem is natural, but standard solution techniques exhibit limitations. The weighted-sum method, while simple, requires careful, often heuristic, selection of weighting factors and may not explore the full Pareto front effectively \cite{shen_2024_waveform, yu_2024_waveform}. The $\epsilon$-constraint method offers a more structured approach by optimizing one objective subject to constraints on the others \cite{jiu_2015_wideband, hao_2019_efficient}, but can still struggle to characterize the complete set of optimal trade-offs. More recently, learning-based approaches, such as deep reinforcement learning (DRL), have shown promise for adaptive waveform selection \cite{thornton_2020_Deep}, yet concerns regarding training complexity, stability, and generalization persist.

Information-theoretic measures offer a potentially more principled avenue for unifying detection and estimation goals. For instance, Xiao \cite{xiao_2022_waveform} utilized a combination of MI and KL divergence in a two-stage adaptive scheme, demonstrating empirical benefits but lacking a rigorous foundational framework for their combination. This paper aims to bridge this gap by proposing a novel joint optimization framework rooted in a variational perspective on mutual information. Specifically, we leverage the decomposition of an expected KL divergence term, often considered as a Mutual Information Upper Bound (MIUB), into the sum of the conventional MI (related to estimation) and the KL divergence between the marginal likelihoods under the two hypotheses (related to detection). Maximizing this MIUB thus provides a unified objective that inherently balances both tasks without resorting to arbitrary weights.

The main contributions of this paper are threefold:
\begin{itemize}
	\item \textbf{Unified Variational Formulation:} We introduce a principled joint detection-estimation optimization framework based on maximizing the MIUB. This approach naturally integrates KL divergence (for detection enhancement) and mutual information (for estimation accuracy) derived from a single information-theoretic quantity, avoiding ad-hoc scalarization.
	\item \textbf{Flexible Statistical Modeling via GMDs:} We employ Gaussian Mixture Distributions (GMDs) to statistically model potentially complex target impulse responses and clutter environments. GMDs offer significant representational flexibility while permitting analytically tractable approximations of the MI and KL divergence terms within the MIUB objective via structured covariance matrix analysis.
	\item \textbf{Phase-Coded Dream Optimization Algorithm (PC-DOA):} We develop a specialized optimization algorithm, PC-DOA, designed to efficiently solve the resultant highly non-convex waveform design problem under practical constant modulus constraints. PC-DOA features a hybrid initialization strategy combining domain knowledge with diverse sampling techniques and employs adaptive exploration-exploitation mechanisms tailored for optimizing phase sequences, enhancing convergence robustness and solution quality.
\end{itemize}

The remainder of this paper is organized as follows: Section~\ref{sec2} details the signal model incorporating GMDs and formally derives the MIUB-based optimization problem formulation. Section~\ref{sec3} presents the proposed Phase-Coded Dream Optimization Algorithm (PC-DOA), including its phase representation, hybrid initialization, and two stage optimization strategy. Section~\ref{sec4} provides some numerical results evaluating the performance of the proposed method against benchmark techniques and analyzing the characteristics of the optimized waveforms. Finally, Section~\ref{sec5} concludes the paper and discusses potential directions for future research.

\textit{Notations:} Bold lowercase and uppercase letters denote vectors and matrices, respectively. The operators include $(\cdot)^{\mathrm{T}}$ for transpose, $(\cdot)^{\mathrm{H}}$ for conjugate transpose, $\ast$ for convolution, $\mathrm{tr}(\cdot)$ for trace, $\det(\cdot)$ for determinant, and $\|\cdot\|_{\mathrm{F}}$ for the Frobenius norm. The $n$-dimensional complex space is represented by $\mathbb{C}^n$. The Kullback-Leibler divergence and mutual information are denoted as $\mathcal{D}_{\text{KL}}(p \parallel q)$ and $I(\mathbf{x}; \mathbf{y})$, respectively. The signal model follows hypotheses $\mathcal{H}_0$ for clutter-only and $\mathcal{H}_1$ for target-plus-clutter. GMD parameters include mixture weights $\alpha_k$, $\beta_m$, and $\gamma_\ell$, as well as covariance matrices $\mathbf{R}_k$ for clutter, $\mathbf{Q}_m$ for the target, and $\bm{\Sigma}_\ell$ for the composite response under $\mathcal{H}_1$. The imaginary unit is $j = \sqrt{-1}$, and “s.t.” denotes “subject to”. $\mathcal{CN}(\bm{\mu},\mathbf{R})$ denotes a complex Gaussian distribution with mean $\bm{\mu}\in\mathbb{C}^n$ and covariance matrix $\mathbf{R}$, whose probability density function is given by
\begin{equation*}
	\mathcal{CN}(\bm{\mu},\mathbf{R})=\frac{\exp\left[-(\mathbf{x}-\bm{\mu})^{\mathrm{H}}\mathbf{R}^{-1}(\mathbf{x}-\bm{\mu})\right]}{\pi^n \det(\mathbf{R})}
\end{equation*}

\section{Signal Model and Problem Formulation}
\label{sec2}
This section establishes the mathematical framework for the radar system, detailing the signal model and formulating the waveform optimization problem for joint enhancement of target detection and parameter estimation.
\subsection{Signal Model}
We consider a monostatic radar transmitting a baseband equivalent waveform $s(t)$. The received signal $y(t)$ comprises target returns potentially corrupted by clutter and noise. Assuming a linear time-invariant target interaction channel, the continuous-time received signal is
\begin{equation}
	y(t) = s(t) \ast x(t) + w(t),
\end{equation}
where $x(t)$ is the stochastic target impulse response (TIR) characterizing target scattering, and $w(t)$ is the composite clutter and noise process. The TIR $x(t)$ and process $w(t)$ are assumed statistically independent of each other and of $s(t)$.
For digital processing, we use a discrete-time representation. The transmitted waveform is $\mathbf{s} = [s_{1}, \dots, s_{N}]^{\mathrm{T}} \in \mathbb{C}^{N}$ ($N$ samples/code length). The discrete TIR is $\mathbf{x} = [x_{1}, \dots, x_{N_{T}}]^{\mathrm{T}} \in \mathbb{C}^{N_{T}}$ ($N_{T}$ effective temporal support). The discrete clutter/noise is $\mathbf{w} = [w_{1}, \dots, w_{N+N_{T}-1}]^{\mathrm{T}} \in \mathbb{C}^{N+N_{T}-1}$.
The target detection task is formulated as a binary hypothesis test on the received data $\mathbf{y} \in \mathbb{C}^{N+N_{T}-1}$:
\begin{equation}
	\begin{cases}
		\mathcal{H}_{0}: & \mathbf{y} = \mathbf{w}  \\
		\mathcal{H}_{1}: & \mathbf{y} = \mathbf{Sx} + \mathbf{w}
	\end{cases}
	\label{eq:hypothesis_test}
\end{equation}
where $\mathbf{S} \in \mathbb{C}^{(N+N_{T}-1) \times N_{T}}$ is the Toeplitz convolution matrix derived from $\mathbf{s}$:
\begin{equation}
	[\mathbf{S}]_{i,j} = 
	\begin{cases}
		s_{i-j+1}, & \text{if } 1 \le j \le N_T, j \le i \le j+N-1 \\
		0, & \text{otherwise}.
	\end{cases}
	\label{eq:conv_matrix}
\end{equation}

To capture complex, potentially non-Gaussian statistics of realistic target scattering and clutter, we model both $\mathbf{x}$ and $\mathbf{w}$ using Gaussian Mixture Distributions (GMDs)~\cite{chen_joint_2023,gu2019information}. GMDs offer flexibility in density approximation and analytical tractability. Assuming zero-mean processes (typical after coherent processing), the models are:
\begin{align}
	p(\mathbf{w}) &= \sum_{k=1}^{K} \alpha_k \mathcal{CN}(\mathbf{w}; \mathbf{0}, \mathbf{R}_k), \label{eq4}\\
	p(\mathbf{x}) &= \sum_{m=1}^{M} \beta_m \mathcal{CN}(\mathbf{x}; \mathbf{0}, \mathbf{Q}_m), \label{eq5}
\end{align}
where $\mathcal{CN}(\cdot; \bm{\mu}, \mathbf{\Sigma})$ is the complex Gaussian PDF. Mixture weights satisfy $\alpha_k > 0$, $\sum_{k}\alpha_k=1$, $\beta_m > 0$, $\sum_{m}\beta_m=1$. Covariance matrices $\{\mathbf{R}_k\}$ and $\{\mathbf{Q}_m\}$ encapsulate second-order statistics (e.g., power spectra) of the components.
Under the GMD assumption, the likelihood functions $p(\mathbf{y} | \mathcal{H}_i)$ are also GMDs:
\begin{align}
	p_{0}(\mathbf{y}) &\triangleq p(\mathbf{y} | \mathcal{H}_0) = \sum_{k=1}^{K}\alpha_k \mathcal{CN}(\mathbf{y}; \mathbf{0}, \mathbf{R}_k), \label{eq:p0_gmd}\\ 
	p_{1}(\mathbf{y}) &\triangleq p(\mathbf{y} | \mathcal{H}_1) = \sum_{\ell=1}^{L}\gamma_\ell \mathcal{CN}(\mathbf{y}; \mathbf{0}, \mathbf{\Sigma}_\ell), \label{eq:p1_gmd}
\end{align}
where $L = MK$, $\gamma_\ell = \alpha_k \beta_m$ for the composite index $\ell \leftrightarrow (k,m)$, and the $\mathcal{H}_1$ component covariance is
\begin{equation}
	\mathbf{\Sigma}_\ell = \mathbf{S}\mathbf{Q}_{m}\mathbf{S}^{\mathrm{H}} + \mathbf{R}_{k}, \quad (\text{for } \ell \leftrightarrow (k,m)).
	\label{eq:Sigma_ell}
\end{equation}

Note the dependence of $\mathbf{\Sigma}_\ell$ on the waveform $\mathbf{s}$ via $\mathbf{S}$.

\begin{remark}
GMDs offer significant representational flexibility, serving as universal approximators capable of fitting arbitrary smooth distributions given sufficient components. Their use is supported by established practice in waveform design, demonstrating excellent performance for modeling measured and simulated target/clutter data \cite{chen_joint_2023, gu2019information}. Furthermore, quantitative analyses show that GMD approximation errors for challenging distributions like K-distributions can yield small deviations in the first- and second-order statistics crucial for detection/estimation \cite{Blacknell_Target_2000}. The GMD structure permits analytically tractable approximations in some specific cases although the potential computational complexity would increase with more components.
\end{remark}

\subsection{Problem Formulation}
\label{subsec:problem_Formulation}

The main objective is to design the waveform $\mathbf{s} \in \mathbb{C}^N$ to jointly optimize detection and estimation performance under the constant modulus constraint. Let $E_s$ denote the fixed total transmitted energy. The constraint is explicitly $|s_n| = c = \sqrt{E_s/N}$ for $n=1, \dots, N$.
A direct multi-objective formulation considers maximizing metrics for detection, $\mathfrak{D}(\mathbf{s})$, and estimation, $\mathfrak{E}(\mathbf{s})$, simultaneously:
\begin{equation}
	(\mathcal{P}_1): \quad 
	\begin{alignedat}{2} 
		&\max_{\mathbf{s}} \quad && [\mathfrak{D}(\mathbf{s}), \mathfrak{E}(\mathbf{s})] \\ 
		&\text{s.t.}       && |s_n| = c, \quad n = 1, \dots, N.
	\end{alignedat}
	\label{eq:P1_concise}
\end{equation}

As previously discussed, resolving trade-offs using standard scalarization techniques may lead to suboptimal or ill-posed solutions. Within the framework of information theory, certain information-theoretic quantities can be rigorously derived to characterize a range of performance metrics. Motivated by the metrics in \cite{pmlr-v97-poole19a}, we adopt the Mutual Information Upper Bound (MIUB) to provide a unified criterion. The MIUB is formally defined as the expected Kullback–Leibler (KL) divergence:
\begin{equation}
	\text{MIUB}(\mathbf{s}) \triangleq \mathbb{E}_{\mathbf{x}}\left[\mathcal{D}_{\mathrm{KL}}\left(p(\mathbf{y}|\mathbf{x}, \mathcal{H}_1) \parallel p_0(\mathbf{y})\right)\right].
	\label{eq:miub_def_concise}
\end{equation}

Its significance lies in the decomposition established by the following lemma.
\begin{lemma}[MIUB Decomposition] \label{lemma:miub_decomp_concise}
	The MIUB decomposes into the sum of the mutual information and the KL divergence between marginal likelihoods:
	\begin{equation}
		\text{MIUB}(\mathbf{s}) = \underbrace{I(\mathbf{x};\mathbf{y})}_{\mathfrak{E}(\mathbf{s})} + \underbrace{\mathcal{D}_{\mathrm{KL}}\left(p_1(\mathbf{y}) \parallel p_0(\mathbf{y})\right)}_{\mathfrak{D}(\mathbf{s}) }.
		\label{eq:miub_identity_concise}
	\end{equation}
\end{lemma}

Based on this Lemma, where $I(\mathbf{x};\mathbf{y})$ relates to estimation and $\mathcal{D}_{\mathrm{KL}}(p_1 \parallel p_0)$ relates to detection \cite{bell_information_1993, cover_elements_nodate}, maximizing the MIUB naturally balances both tasks. This leads to the single-objective problem:
\begin{equation}
	(\mathcal{P}_2): \quad
	\begin{alignedat}{2}
		&\max_{\mathbf{s}} \quad && \text{MIUB}(\mathbf{s}) \\
		&\text{s.t.}       && |s_n|=c, \quad n = 1, \dots, N.
	\end{alignedat}
	\label{eq:P2_concise}
\end{equation}

Directly computing the MI and KL divergence terms in \eqref{eq:miub_identity_concise} using the GMD likelihoods $p_0(\mathbf{y})$ in \eqref{eq:p0_gmd} and $p_1(\mathbf{y})$ in \eqref{eq:p1_gmd} is generally intractable. We therefore consider the following approximations.
\begin{lemma}[MI Approximation~\cite{gu2019information}] \label{lemma2_concise}
	The mutual information term $I(\mathbf{x};\mathbf{y})$ is approximated by $\overline{\mathfrak{E}}(\mathbf{s})$:
	\begin{align}
		\MoveEqLeft[2] \overline{\mathfrak{E}}(\mathbf{s}) 
		= \log\left[\sum_{k=1}^{K}\alpha_k\det\left(\mathbf{R}_k\right)^{-1}\right] \nonumber \\ 
		&\quad - \log\left[\sum_{\ell=1}^{L}\gamma_{\ell}\det\left(\mathbf{\Sigma}_{\ell}\right)^{-1}\right].
		\label{eq15_concise}
	\end{align}
\end{lemma}

\begin{lemma}[KL Divergence Approximation \cite{goldbergerEfficient2003}] \label{lemma3_concise}
	The KL divergence $\mathcal{D}_{\mathrm{KL}}(p_1 \parallel p_0)$ is approximated via component-wise matching by $\overline{\mathfrak{D}}(\mathbf{s})$:
	\begin{align}
		\MoveEqLeft[2] \overline{\mathfrak{D}}(\mathbf{s}) =\sum_{\ell=1}^{L}\gamma_{\ell}\Biggl[\log\frac{\gamma_{\ell}}{\alpha_{k^\star(\ell)}} \nonumber \\
		&+\mathcal{D}_{\mathrm{KL}}\Bigl[\mathcal{CN}(\mathbf{0},\mathbf{\Sigma}_{\ell})\parallel \mathcal{CN}(\mathbf{0},\mathbf{R}_{k^\star(\ell)})\Bigr]\Biggr],
		\label{eq16_concise}
	\end{align}
	where $k^\star(\ell) = \arg\min_{k \in \{1,\dots,K\}} J(k, \ell)$, with
	\begin{equation}
		J(k, \ell) = \log\frac{\gamma_{\ell}}{\alpha_k}+\mathcal{D}_{\mathrm{KL}}\Bigl[\mathcal{CN}(\mathbf{0},\mathbf{\Sigma}_{\ell})\parallel \mathcal{CN}(\mathbf{0},\mathbf{R}_k)\Bigr].
		\label{eq:J_k_ell_concise}
	\end{equation}
	
    The KL divergence between the zero-mean complex Gaussian components is:
	\begin{align}
		\MoveEqLeft[2] \mathcal{D}_{\mathrm{KL}}\Bigl[\mathcal{CN}(\mathbf{0},\mathbf{A})\parallel \mathcal{CN}(\mathbf{0},\mathbf{B})\Bigr] \nonumber \\ 
		&= \operatorname{tr}\left(\mathbf{B}^{-1}\mathbf{A}\right)-\log\det\left(\mathbf{B}^{-1}\mathbf{A}\right)-n',
		\label{eq:KL_gaussian_concise}
	\end{align}
	with $n' = N+N_T-1$.
\end{lemma}

\begin{remark}
A brief error analysis for the approximations in Lemmas~\ref{lemma2_concise} and~\ref{lemma3_concise} is provided in Appendix~\ref{app:error_analysis}. This analysis establishes that the surrogate formulations are theoretically feasible with controllable approximation errors under some specific assumptions (i.e. low variance for MI, well-separated components for KL).
\end{remark}

Substituting these approximations into $(\mathcal{P}_2)$ yields the final, computationally tractable optimization problem:
\begin{equation}\label{eq17_concise}
	(\mathcal{P}_3): \quad
	\begin{alignedat}{2} 
		&\max_{\mathbf{s}} \quad && F(\mathbf{s}) \triangleq \overline{\mathfrak{D}}(\mathbf{s})+\overline{\mathfrak{E}}(\mathbf{s})\\
		&\text{s.t.}      && |s_n|=c, \quad n = 1, \dots, N.
	\end{alignedat}
\end{equation}

It should be noted that the objective function $F(\mathbf{s})$ defined in $(\mathcal{P}_3)$ is Lipschitz continuity according to the following proposition.

\begin{proposition}\label{prop1}
	Assume that the component covariance matrices $\{\mathbf{R}_k\}_{k=1}^K$ and $\{\mathbf{Q}_m\}_{m=1}^M$ are bounded (i.e., their eigenvalues lie within $[\lambda_{\min}, \lambda_{\max}]$ with $0 < \lambda_{\min} \le \lambda_{\max} < \infty$) and that the mixture weights are strictly positive ($\min\{\alpha_k, \beta_m, \gamma_\ell\} > 0$). Under the constant modulus constraint $|s_n|=c$ for all $n$, the objective function $F(\mathbf{s})$ is Lipschitz continuous on the feasible set $\mathcal{M}$. There exists a constant $L_{\mathrm{F}} < \infty$ such that
	\begin{equation}
		|F(\mathbf{s}_1) - F(\mathbf{s}_2)| \leq L_{\mathrm{F}} \|\mathbf{s}_1 - \mathbf{s}_2\|_2, \quad \forall \mathbf{s}_1, \mathbf{s}_2 \in \mathcal{M}.
		\label{eq:lipschitz}
	\end{equation}
	The Lipschitz constant $L_{\mathrm{F}}$ depends on the bounds $\lambda_{\min}$, $\lambda_{\max}$, the waveform dimension $N$, the minimum mixture weight, the constant $c$ (related to total energy $E_s$), and the dimensions $N_T$, $K$, $M$.
\end{proposition}

\begin{proof} 
 See \ref{sec:lipschitz_proof} for details.
\end{proof}

The objective function $F(\mathbf{s})$ involves computable matrix operations dependent on $\mathbf{s}$ via $\mathbf{\Sigma}_\ell$ \eqref{eq:Sigma_ell}. However, the non-convexity of both $F(\mathbf{s})$ and the constant modulus feasible set necessitates the specialized optimization approach detailed in Section~\ref{sec3}.

\section{Proposed Phase-Coded Dream Optimization Algorithm}
\label{sec3} 
The optimization problem $(\mathcal{P}_3)$ in \eqref{eq17_concise} presents considerable challenges owing to the non-convex nature of the objective function $F(\mathbf{s})$, compounded by the non-convex geometry of the constant modulus feasible set $\mathcal{M}$. Specifically, the presence of logarithmic determinant terms in $\mathfrak{D}(\mathbf{s})$ and $\mathfrak{E}(\mathbf{s})$ makes the landscape highly complex. Conventional gradient-based optimization techniques are often ill-suited for such landscapes, as gradients may be difficult or computationally expensive to obtain and are susceptible to convergence towards potentially numerous suboptimal local extrema. To surmount these obstacles, this section introduces a tailored metaheuristic approach, the Phase-Coded Dream Optimization Algorithm (PC-DOA), specifically architected for navigating the complexities inherent in constant modulus radar waveform design based on the MIUB criterion. PC-DOA eschews reliance on gradient information, leveraging instead a sophisticated population-based search strategy enhanced with mechanisms attuned to the problem structure, including a domain-specific hybrid initialization and adaptive phase updates.

\subsection{Phase-Coded Representation}
\label{subsec:representation}
A cornerstone of the proposed methodology is the direct enforcement of the constant modulus constraint $|s_n| = c$ by reformulating the optimization problem on the phase space. The feasible set constitutes the complex circle manifold:
\begin{equation}
	\mathcal{M} = \left\{\mathbf{s} \in \mathbb{C}^N : |s_n| = c,~\forall n=1, \dots, N \right\}.
	\label{eq:manifold_def}
\end{equation}

Recalling $c = \sqrt{E_s/N}$, any waveform $\mathbf{s} \in \mathcal{M}$ can be uniquely parameterized by its phase vector $\bm{\theta} = [\theta_1, \dots, \theta_N]^{\mathrm{T}} \in [-\pi, \pi)^N$ (or any interval of length $2\pi$):
\begin{equation}
	\mathbf{s}(\bm{\theta}) = c \cdot [e^{j\theta_1}, \dots, e^{j\theta_N}]^{\mathrm{T}}.
	\label{eq:phase_parameterization}
\end{equation}

This parameterization inherently satisfies the constant modulus requirement. Optimization is thus performed over the phase vector $\bm{\theta}$. To ensure phase values generated during updates remain within the principal interval $[-\pi, \pi)$, we employ a periodic wrapping operator $\Psi: \mathbb{R}^N \to [-\pi, \pi)^N$, defined element-wise as:
\begin{equation}
	[\Psi(\bm{\theta})]_n = \Psiop(\theta_n + \pi, 2\pi) - \pi,
	\label{eq:phase_wrap_op} 
\end{equation}
where $\text{mod}(a, b)$ denotes the modulo operation yielding a result in $[0, b)$. Any phase vector resulting from algorithmic updates is passed through $\Psi(\cdot)$. This strategy integrates the constraint into the search space representation, facilitating navigation of the solution manifold without explicit projection steps.
\subsection{Hybrid Initialization Strategy}
\label{subsec:initialization}
The efficacy of population-based metaheuristics like PC-DOA is profoundly influenced by the diversity and quality of the initial population. To accelerate convergence and mitigate premature stagnation, we propose a hybrid initialization strategy. Let the total population size be $N_p$. The procedure generates $N_p$ initial phase vectors $\{\bm{\theta}_i^{(0)}\}_{i=1}^{N_p}$ using three complementary mechanisms:

\subsubsection{LFM-Inspired Initialization}

Capitalizing on the favorable properties of Linear Frequency Modulation (LFM) waveforms, a fraction $N_{\text{LFM}} = \lfloor \eta N_p \rfloor$ of the population is initialized with phase structures reminiscent of LFM signals, where $\eta \in (0,1)$ is a control parameter. The $n$-th phase component for the $i$-th such individual ($i=1, \dots, N_{\text{LFM}}$) is:
\begin{equation}
	[\bm{\theta}_{\text{LFM},i}^{(0)}]_n = \Psi\left( \beta_i \pi \frac{(n-1)^2}{N-1} + \delta_{i,n} \right), \quad n=1,\dots,N,
	\label{eq:lfm_init_revised} 
\end{equation}
where $\beta_i$ is a randomly sampled chirp rate parameter (e.g., $\beta_i \sim \mathcal{U}(-1, 1)$) and $\delta_{i,n} \sim \mathcal{U}(-\Delta, \Delta)$ is a small random phase perturbation with $\Delta \ll \pi$. This seeds the population with potentially promising structures. $\mathcal{U}(a, b)$ denotes the continuous uniform distribution over $[a, b)$.

\subsubsection{Chaotic Initialization}

To leverage the ergodicity of chaotic systems for broad exploration, $N_{\text{chaos}} = \lfloor (N_p - N_{\text{LFM}})/2 \rfloor$ individuals are initialized using the Logistic map. For each such individual $i$, an initial vector $\mathbf{c}^{(0)} \in (0, 1)^N$ is randomly generated (excluding fixed points). It is iterated $K_{\text{chaos}}$ times via:
\begin{equation}
	c_n^{(k)} = 4 c_n^{(k-1)} (1 - c_n^{(k-1)}),
	\label{eq:chaos_map_revised} 
\end{equation}
where $n=1,\dots,N$, $k=1,\dots,K_{\text{chaos}}$.

The resulting vector $\mathbf{c}^{(K_{\text{chaos}})}$ maps to the phase domain:
\begin{equation}
	[\bm{\theta}_{\text{chaos},i}^{(0)}]_n = \Psi\left( -\pi + 2\pi c_n^{(K_{\text{chaos}})} \right).
	\label{eq:chaos_init_revised} 
\end{equation}

This promotes a thorough initial canvassing of the phase space $\mathbb{T}^N \cong [-\pi, \pi)^N$.

\subsubsection{Random Initialization}

To ensure unbiased coverage, the remaining $N_{\text{rand}} = N_p - N_{\text{LFM}} - N_{\text{chaos}}$ individuals are initialized purely randomly:
\begin{equation}
	[\bm{\theta}_{\text{rand},i}^{(0)}]_n \sim \mathcal{U}(-\pi, \pi), \quad n=1,\dots,N.
	\label{eq:rand_init_revised} 
\end{equation}

This enhances the exploratory capability of the initial population.
This tripartite strategy provides a robust starting point, balancing exploitation of known structures with diverse exploration.

\subsection{Exploration and Exploitation}
\label{subsec:optimization}

The core of PC-DOA is its iterative refinement process, balancing exploration and exploitation. Let $\mathcal{P}_t = \{\bm{\theta}_1^{(t)}, \dots, \bm{\theta}_{N_p}^{(t)}\}$ be the population at iteration $t$. The fitness of each candidate $\bm{\theta}_i^{(t)}$ is evaluated using the objective function from $\mathcal{P}_3$:
\begin{equation}
	f(\bm{\theta}_i^{(t)}) = F(\mathbf{s}(\bm{\theta}_i^{(t)})) = \overline{\mathfrak{D}}(\mathbf{s}(\bm{\theta}_i^{(t)})) + \overline{\mathfrak{E}}(\mathbf{s}(\bm{\theta}_i^{(t)})),
	\label{eq:fitness_obj_revised} 
\end{equation}
with $\overline{\mathfrak{D}}(\cdot)$ and $\overline{\mathfrak{E}}(\cdot)$ given by \eqref{eq16_concise} and \eqref{eq15_concise}, respectively, and $\mathbf{s}(\bm{\theta})$ by \eqref{eq:phase_parameterization}. The optimization proceeds in two phases, controlled by the iteration count $t$ relative to the maximum iterations $T$.
\subsubsection{Exploration Phase ($t \leq \alpha T$)}
For early iterations ($t \leq \alpha T$, e.g., $\alpha=0.9$), the focus is on broad exploration. The population $\mathcal{P}_{t-1}$ is partitioned into $G$ subgroups $\mathcal{G}_1, \dots, \mathcal{G}_G$. Within each subgroup $\mathcal{G}_g$, the local elite (best individual) is identified:
\begin{equation}
	\bm{\theta}_{g}^{\text{best}} = \argmax_{\bm{\theta} \in \mathcal{G}_g} f(\bm{\theta}).
	\label{eq:local_elite}
\end{equation}

Individuals are updated incorporating stochastic perturbations. For an individual $\bm{\theta}_{\text{old}} \in \mathcal{G}_g$, a subset $\mathcal{I}$ of $k$ dimensions is randomly selected, where $k \sim \mathcal{U}_{\text{int}}([\lceil N/(8G)\rceil, \lceil N/(3G)\rceil])$. ($\mathcal{U}_{\text{int}}$ denotes discrete uniform distribution). The update for $n \in \mathcal{I}$ is:
\begin{equation}
	[\bm{\theta}_{\text{new}}]_n = \Psi\Bigl([\bm{\theta}_{\text{old}}]_n + \zeta(t) \cdot \delta_n \Bigr), \quad n \in \mathcal{I},
	\label{eq:explore_update_revised} 
\end{equation}
where $\delta_n \sim \mathcal{U}(-\pi, \pi)$ is a random perturbation, and $\zeta(t)$ is a decaying scale factor allowing larger steps early on, e.g., via cosine annealing:
\begin{equation}
	\zeta(t) = \frac{1}{2}\left(\cos\left(\frac{\pi t}{T}\right) + 1\right).
	\label{eq:zeta_decay_revised} 
\end{equation}

For $n \notin \mathcal{I}$, $[\bm{\theta}_{\text{new}}]_n = [\bm{\theta}_{\text{old}}]_n$. This subgroup-based exploration maintains diversity.
\subsubsection{Exploitation Phase ($t > \alpha T$)}
As optimization progresses ($t > \alpha T$), focus shifts to local refinement. The global best solution found so far, $\bm{\theta}^{\text{global}}$, guides the search. Updates involve smaller perturbations. For an individual $\bm{\theta}_{\text{old}}$, a subset $\mathcal{I}$ of $k$ dimensions is randomly selected, with $k \leq k_{\max} = \max(2, \lceil N/3\rceil)$. The update for $n \in \mathcal{I}$ is:
\begin{equation}
	[\bm{\theta}_{\text{new}}]_n = \Psi\Bigl([\bm{\theta}^{\text{global}}]_n + \eta(t) \cdot \delta_n \Bigr), \quad n \in \mathcal{I},
	\label{eq:exploit_update_revised} 
\end{equation}
where $\delta_n \sim \mathcal{U}(-\pi, \pi)$ and $\eta(t)$ is an adaptive step-size factor, typically smaller than $\zeta(t)$ in later stages. Following the provided text structure, we set:
\begin{equation}
	\eta(t) = \frac{1}{2}\left(\cos\left(\frac{\pi t}{T}\right) + 1\right).
	\label{eq:eta_decay_revised} 
\end{equation}

For $n \notin \mathcal{I}$, $[\bm{\theta}_{\text{new}}]_n = [\bm{\theta}_{\text{old}}]_n$. This promotes convergence towards high-quality solutions. 
\subsubsection{Implementation Summary}
The overall PC-DOA iterates these phases, potentially incorporating elite preservation (carrying forward the best individuals). The procedure is outlined in Algorithm~\ref{alg:pc_doa_revised}.
\begin{algorithm}[!t]
	\caption{Phase-Coded Dream Optimization Algorithm (PC-DOA)}
	\label{alg:pc_doa_revised} 
	\begin{algorithmic}[1]
		\REQUIRE Waveform dimension $N$, population size $N_p$, max iterations $T$, constant magnitude $c$, initialization proportion $\eta$, phase transition parameter $\alpha$, number of groups $G$.
		\ENSURE Optimal phase vector $\bm{\theta}^*$ and corresponding waveform $\mathbf{s}^*$.
		
		\STATE Initialize population $\mathcal{P}_0 = \{\bm{\theta}_1^{(0)}, \dots, \bm{\theta}_{N_p}^{(0)}\}$ using hybrid strategy \eqref{eq:lfm_init_revised}-\eqref{eq:rand_init_revised}.
		\STATE Evaluate $f(\bm{\theta}_i^{(0)})$ for all $i \in \{1, \dots, N_p\}$ using \eqref{eq:fitness_obj_revised}.
		\STATE Set $\bm{\theta}^{\text{global}} = \argmax_{i} f(\bm{\theta}_i^{(0)})$ and $f^{\text{global}} = f(\bm{\theta}^{\text{global}})$.
		
		\FOR{$t = 1$ to $T$}
		\STATE $\mathcal{P}_{\text{new}} = \emptyset$.
		\IF{$t \leq \alpha T$}
		\STATE Partition $\mathcal{P}_{t-1}$ into $G$ subgroups $\mathcal{G}_1, \dots, \mathcal{G}_G$.
		\FOR{each subgroup $\mathcal{G}_g$}
		\STATE Identify local elite $\bm{\theta}_{g}^{\text{best}}$ using \eqref{eq:local_elite}.
		\FOR{each $\bm{\theta}_{\text{old}} \in \mathcal{G}_g$}
		\STATE Generate $\bm{\theta}_{\text{new}}$ using exploration update \eqref{eq:explore_update_revised} (potentially guided by $\bm{\theta}_{g}^{\text{best}}$ or other members, details depend on exact DOA variant). 
		\STATE $\mathcal{P}_{\text{new}} = \mathcal{P}_{\text{new}} \cup \{\bm{\theta}_{\text{new}}\}$.
		\ENDFOR
		\ENDFOR
		\ELSE
		\FOR{each $\bm{\theta}_{\text{old}} \in \mathcal{P}_{t-1}$}
		\STATE Generate $\bm{\theta}_{\text{new}}$ using exploitation update \eqref{eq:exploit_update_revised} biased by $\bm{\theta}^{\text{global}}$.
		\STATE $\mathcal{P}_{\text{new}} = \mathcal{P}_{\text{new}} \cup \{\bm{\theta}_{\text{new}}\}$.
		\ENDFOR
		\ENDIF
		
		\STATE Evaluate $f(\bm{\theta})$ for all $\bm{\theta} \in \mathcal{P}_{\text{new}}$.
		\STATE Apply selection/replacement to form $\mathcal{P}_t$ from $\mathcal{P}_{t-1}$ and $\mathcal{P}_{\text{new}}$ (e.g., greedy selection, potential elitism).
		\STATE Update $\bm{\theta}^{\text{global}}$ and $f^{\text{global}}$ if a better solution is found in $\mathcal{P}_t$.
		\ENDFOR
		
		\STATE Set $\bm{\theta}^* = \bm{\theta}^{\text{global}}$.
		\STATE Compute $\mathbf{s}^* = c \cdot [e^{j\theta^*_1}, \dots, e^{j\theta^*_N}]^{\mathrm{T}}$.
		\RETURN $\bm{\theta}^*$, $\mathbf{s}^*$.
	\end{algorithmic}
\end{algorithm}
\subsection{Computational Complexity and Convergence Analysis}
\label{subsec:complexity_convergence}
\subsubsection{Computational Complexity}
The computational burden per iteration of PC-DOA is dominated by the fitness evaluation step \eqref{eq:fitness_obj_revised}. Evaluating $F(\mathbf{s}(\bm{\theta}))$ involves calculating $\overline{\mathfrak{D}}$ and $\overline{\mathfrak{E}}$, which require operations on covariance matrices of size $n' \times n'$, where $n' = N + N_T - 1$. Computing the KL divergence and MI approximations \eqref{eq16_concise}, \eqref{eq15_concise} involves matrix determinants, inverses, and traces, scaling roughly as $O(L \cdot (n')^3)$ per candidate, where $L=MK$. Evaluating the fitness for the population $N_p$ times incurs a cost of approximately $O(N_p L (n')^3)$ per iteration. The particle update steps typically have lower complexity (e.g., $O(N_p N)$). Thus, the overall complexity for $T$ iterations is approximately $O(T N_p L (n')^3)$.
\subsubsection{Convergence Properties}

The Lipschitz continuity established in Proposition~\ref{prop1} is instrumental, ensuring a degree of smoothness or regularity in the objective function landscape, which is often a prerequisite for the convergence analysis of optimization algorithms.

Regarding the PC-DOA itself, while a formal proof of global convergence for general metaheuristics on non-convex manifolds is intricate, the algorithm incorporates several features conducive to finding high-quality solutions. The stochastic nature of the updates, particularly during the exploration phase, combined with the mechanisms for maintaining population diversity, allows the algorithm to traverse the search space effectively. If the stochastic update rules ensure sufficient exploration such that any region of the feasible manifold $\mathcal{M}$ can be reached from any other state, and combined with appropriate management of step sizes, the algorithm's trajectory can be analyzed within frameworks related to stochastic approximation or global random search methods \cite{spall2005introduction}. Under such ideal conditions, supported by the regularity guaranteed by Proposition~\ref{prop1}, the algorithm can be expected to converge almost surely to the global optimum. Let $F^* = \sup_{\mathbf{s} \in \mathcal{M}} F(\mathbf{s})$ be the maximum value of the objective function. Then, convergence is characterized by:
\begin{equation}
	\Pr\left( \lim_{t\to \infty} F(\mathbf{s}(\bm{\theta}^{(t)})) = F^{*}\right) = 1,
	\label{eq:almost_sure_conv}
\end{equation}
where $\bm{\theta}^{(t)}$ represents the phase vector of the best solution found up to iteration $t$.

\section{Numerical Results}
\label{sec4}
This section presents numerical simulations designed to validate the effectiveness of the proposed Phase-Coded Dream Optimization Algorithm (PC-DOA) and the associated Mutual Information Upper Bound (MIUB) objective function $F(\mathbf{s})$ for constant modulus waveform design.

\subsection{Simulation Configuration} 
\label{sec4.1}
We consider a baseband transmit waveform $\mathbf{s} \in \mathbb{C}^N$ with length $N=64$ and normalized total energy $E_s = 1$ (i.e. $c = 1/\sqrt{N}$). The TIR $\mathbf{x}$ and the clutter process $\mathbf{w}$ are modeled as zero-mean random vectors drawn from GMDs. A foundational assumption, common in cognitive radar waveform design literature (e.g., \cite{bell_information_1993, romeroTheoryApplicationSNR2011}), is that sufficiently accurate statistical models (here, the GMD parameters: weights $\alpha_k, \beta_m$ and covariance matrices $\mathbf{Q}_m, \mathbf{R}_k$) can be obtained \textit{prior} to waveform design, perhaps through prior sensing or intelligence data. Their covariance matrices, $\mathbf{Q} = \sum_{m=1}^M \beta_m \mathbf{Q}_m$ and $\mathbf{R} = \sum_{k=1}^K \alpha_k \mathbf{R}_k$, respectively, are determined via the inverse Fourier transform of predefined power spectral densities (PSDs). Both target and clutter are modeled with $K=M=2$ components. The target PSD components possess mixture weights $\boldsymbol{\beta} = [0.8, 0.2]^T$. Conversely, the clutter PSD components have weights $\boldsymbol{\alpha} = [0.2, 0.8]^T$. These spectral distributions, illustrated in Fig.~\ref{fig:psd_example}, were generated using the methodology presented in \cite{yang_cognitive_2023}. The Signal-to-Clutter Ratio (SCR) is varied parametrically from -5 dB to 5 dB, defined based on the total power associated with target and clutter components:
\begin{equation}
	\mathrm{SCR} = 10\log_{10}\left( \frac{\sum_{m=1}^2 \beta_m \operatorname{Tr}(\mathbf{Q}_m)}{\sum_{k=1}^2 \alpha_k \operatorname{Tr}(\mathbf{R}_k)} \right),
	\label{eq:scr_definition}
\end{equation}
assuming component covariance matrices $\mathbf{Q}_m, \mathbf{R}_k$ capture the spectral power distribution.

The proposed PC-DOA (Algorithm~\ref{alg:pc_doa_revised}) is parameterized with a population size $N_p = 200$, a maximum iteration count $T = 2000$, and an LFM-inspired initialization proportion $\eta = 0.3$. The adaptive step-size factors $\zeta(t)$ and $\eta(t)$ follow the cosine annealing schedules defined in \eqref{eq:zeta_decay_revised} and \eqref{eq:eta_decay_revised}, respectively.
Performance is benchmarked against three strategies:
\begin{enumerate}
	\item \textit{MI Maximization:} Optimization maximizing the mutual information (or a related MI metric $f_{\text{MI}}(\mathbf{s})$), subject to the constant modulus constraint ($\mathbf{s} \in \mathcal{M}$), solved using PC-DOA.
	\item \textit{Weighted SCR-MI (WSM):} Optimization maximizing a weighted sum objective $J_{\text{WSM}}(\mathbf{s}) = w \cdot f_{\text{SCR}}(\mathbf{s}) + (1-w) \cdot f_{\text{MI}}(\mathbf{s})$ with weight $w=0.5$. Here, $f_{\text{SCR}}(\mathbf{s})$ and $f_{\text{MI}}(\mathbf{s})$ represent suitable metrics quantifying the Signal-to-Clutter Ratio and Mutual Information achievable with waveform $\mathbf{s}$, respectively. The optimization is performed directly over the phase variables subject to the constant modulus constraint ($\mathbf{s} \in \mathcal{M}$), using an algorithm such as PC-DOA.
	\item \textit{Random Phase Coding (RPC):} Waveforms $\mathbf{s}$ with phase components $\theta_n$ being drawn independently and uniformly from $[-\pi, \pi)$.
\end{enumerate}

The proposed method optimizes the MIUB objective $F(\mathbf{s})$ in \eqref{eq17_concise} using PC-DOA.
Detection performance is evaluated using Monte Carlo simulations, estimating the detection probability ($P_d$) versus the false alarm probability ($P_{fa}$) by generating Receiver Operating Characteristic (ROC) curves. The number of trials is set to $10^5/P_{fa}$ to ensure statistical significance. Estimation accuracy is assessed by the Mean Squared Error (MSE) of the TIR reconstruction:
\begin{equation}
	\mathrm{MSE} = \mathbb{E}\left[\|\hat{\mathbf{x}} - \mathbf{x}\|_2^2\right],
	\label{eq:mse_definition}
\end{equation}
where $\hat{\mathbf{x}}$ denotes the estimated TIR, typically obtained via a matched filtering or Minimum MSE (MMSE) process dependent on the assumed signal model at the receiver.
\begin{figure}[!t] 
	\centering
	\includegraphics[width=\linewidth]{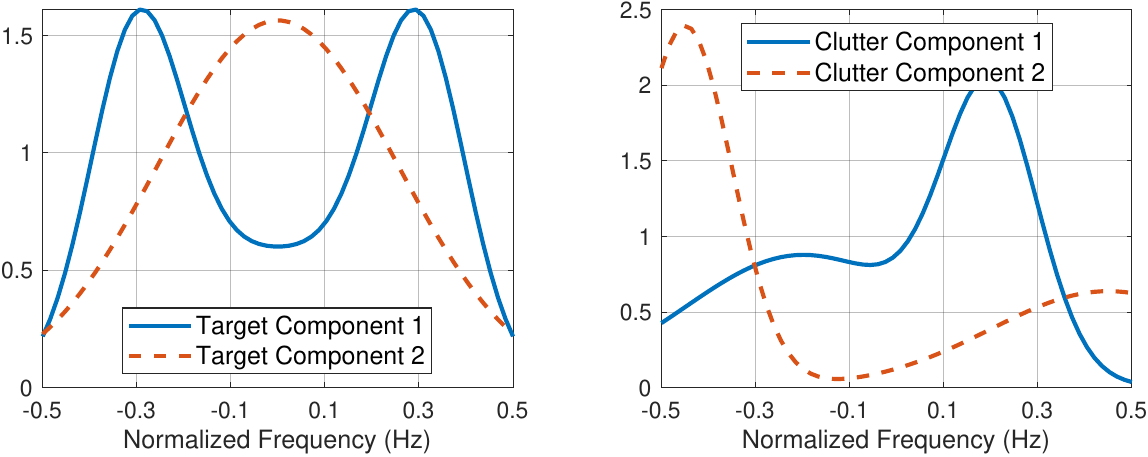} 
	\caption{Illustrative Power Spectral Densities (PSDs) of the Gaussian components comprising the target  and clutter models.} 
	\label{fig:psd_example} 
\end{figure}
\subsection{Performance Analysis}
\label{sec4.2}
\subsubsection{Comparison of Optimization Objectives}
Fig.~\ref{fig:pd_vs_scr} and \ref{fig:roc_curves} compare the detection performance achieved by waveforms optimized under different criteria. Both the proposed MIUB-based and MI-based waveforms demonstrate markedly superior detection capabilities compared to the WSM and RPC approaches across the considered SCR range and for various $P_{fa}$ levels. Notably, optimizing the MIUB objective $F(\mathbf{s})$ yields marginally better detection performance than directly optimizing the MI under these simulation conditions. This observation is complemented by the estimation performance shown in Fig.~\ref{fig:mse_vs_scr}. Here, the MIUB-optimized waveform exhibits a slight increase in MSE compared to the MI-optimized waveform, indicating the expected trade-off between detection and estimation fidelity inherent in waveform design. The close performance between MIUB and MI suggests that, within the constant modulus constraint space, optimizing the upper bound aligns well with optimizing the mutual information itself for estimation purposes.
\begin{figure}[!t]
	\centering
	\includegraphics[width=0.85\linewidth]{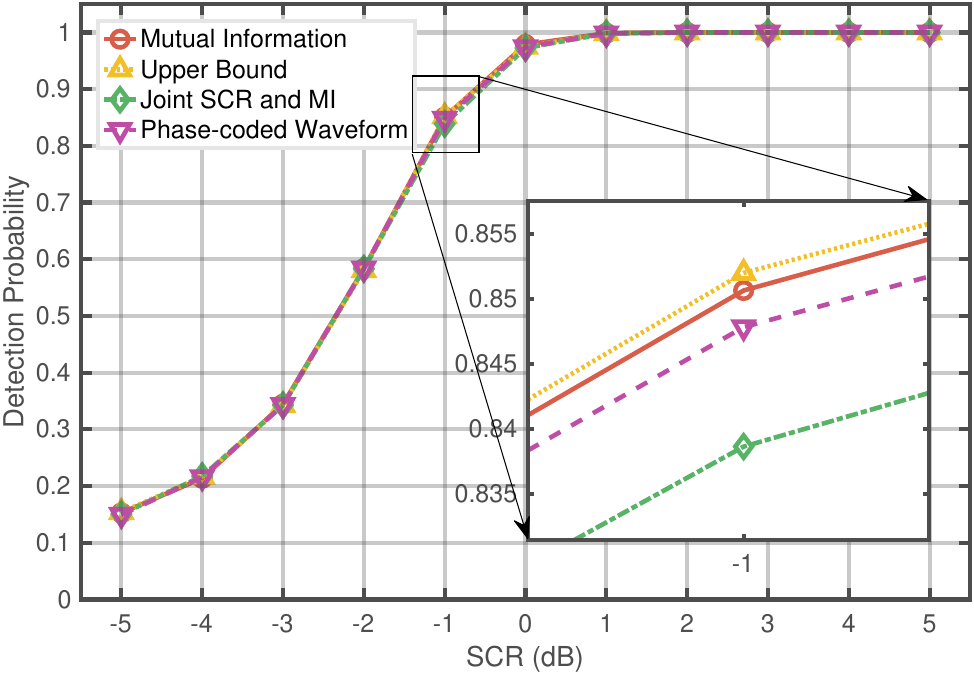}
	\caption{Detection probability ($P_d$) versus Signal-to-Clutter Ratio (SCR) for waveforms optimized using different criteria ($P_{fa} = 10^{-3}$).}
	\label{fig:pd_vs_scr}
\end{figure}
\begin{figure}[!t]
	\centering
	\includegraphics[width=0.85\linewidth]{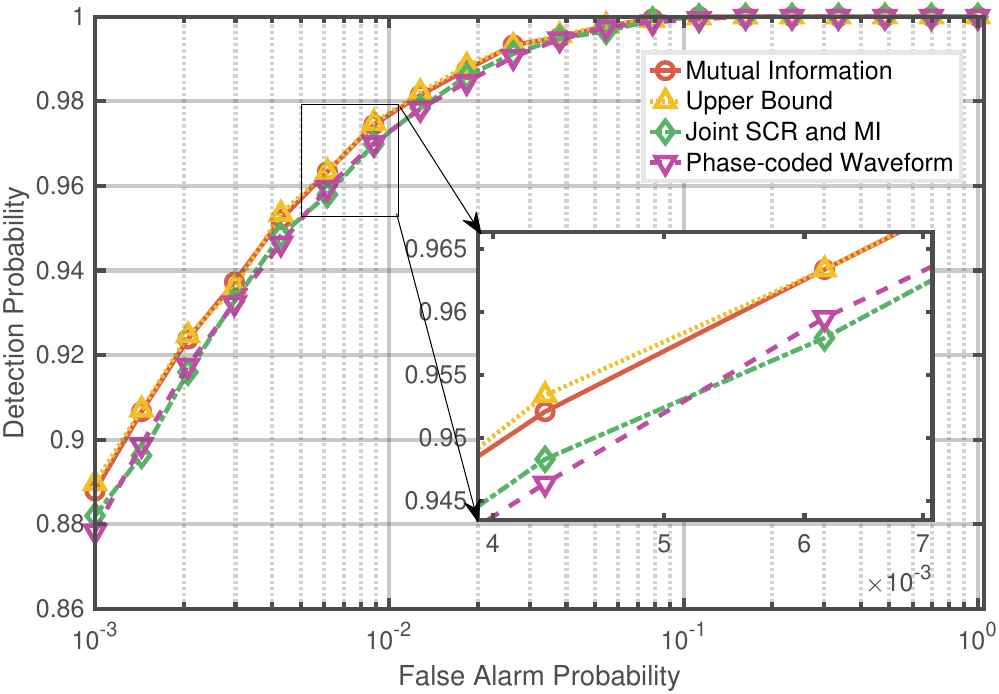}
	\caption{Receiver Operating Characteristic (ROC) curves ($P_d$ vs. $P_{fa}$) for waveforms optimized using different criteria (SCR = 0 dB).}
	\label{fig:roc_curves}
\end{figure}
\begin{figure}[!t]
	\centering
	\includegraphics[width=0.85\linewidth]{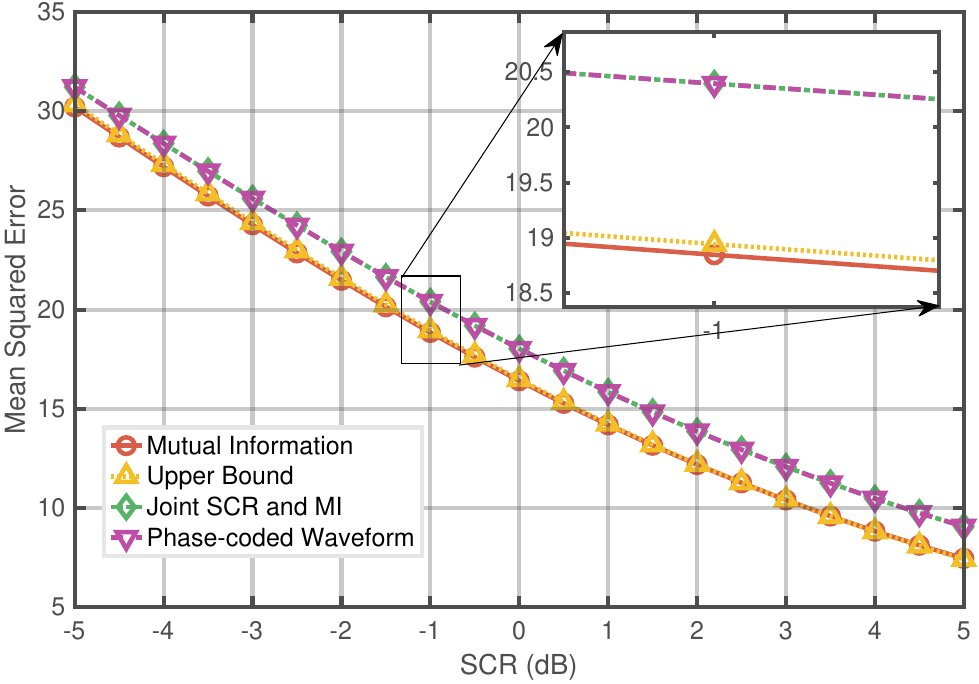}
	\caption{Mean Squared Error (MSE) of Target Impulse Response (TIR) estimation versus SCR for waveforms optimized using different criteria.}
	\label{fig:mse_vs_scr}
\end{figure}
\subsubsection{Comparison of Optimization Algorithms}
We next compare the proposed PC-DOA with a standard Particle Swarm Optimization (PSO) algorithm, both applied to optimize the MIUB objective $F(\mathbf{s})$. Fig.~\ref{fig:roc_doa_pso} and \ref{fig:mse_doa_pso} illustrate the detection and estimation performance. Both PC-DOA and PSO significantly outperform RPC, particularly for SCR values between -2 dB and 4 dB. The ROC curves indicate that PC-DOA achieves enhanced detection, especially under stringent false alarm constraints ($P_{fa} < 10^{-4}$). Regarding estimation, Fig.~\ref{fig:mse_doa_pso} shows that both metaheuristic approaches yield lower MSE compared to RPC, with PC-DOA demonstrating a slight advantage, potentially attributable to its tailored exploration mechanisms on the phase manifold.
\begin{figure}[!t]
	\centering
	\includegraphics[width=0.85\linewidth]{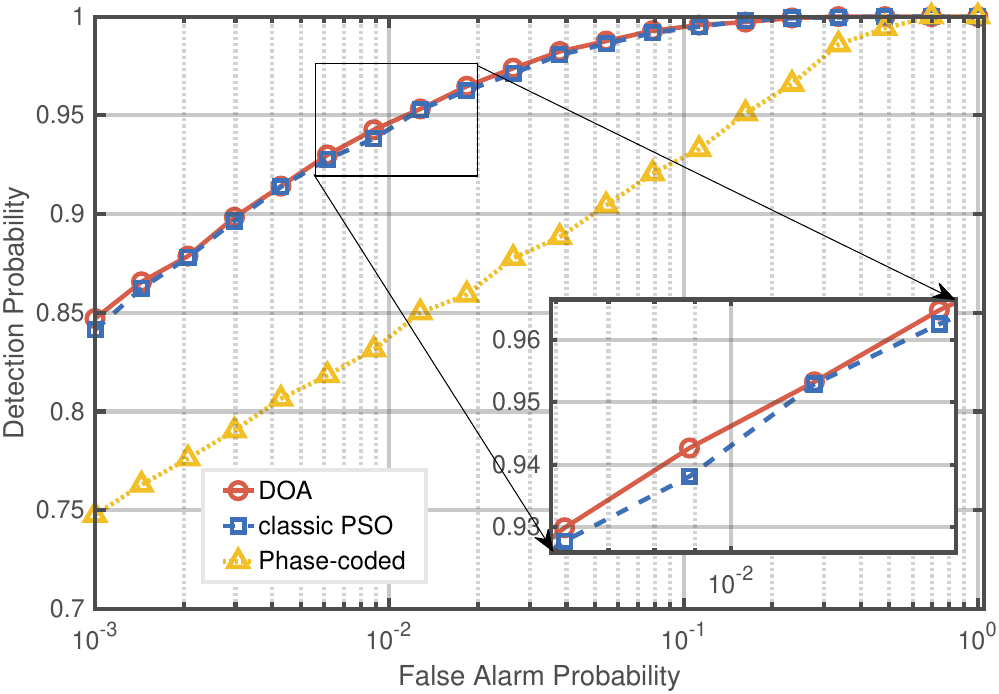}
	\caption{ROC curves comparing PC-DOA and PSO optimization algorithms applied to the MIUB objective (SCR = 0 dB). RPC included for reference.}
	\label{fig:roc_doa_pso}
\end{figure}
\begin{figure}[!t]
	\centering
	\includegraphics[width=0.85\linewidth]{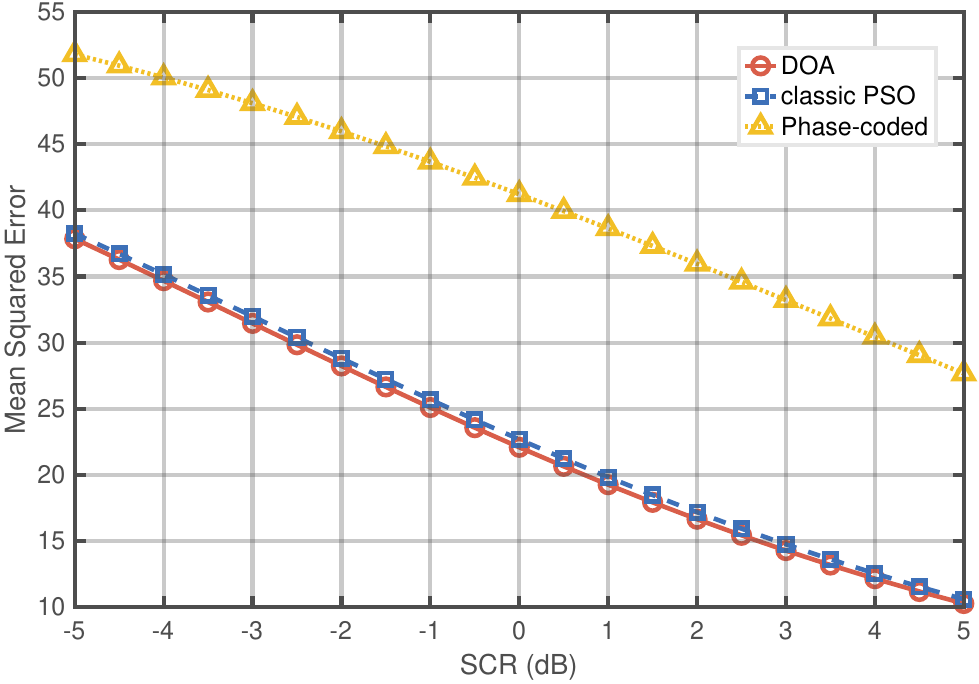}
	\caption{MSE of TIR estimation versus SCR comparing PC-DOA and PSO algorithms applied to the MIUB objective. RPC included for reference.}
	\label{fig:mse_doa_pso}
\end{figure}
Fig.~\ref{fig:convergence} depicts the convergence behavior of PC-DOA and PSO, plotting the best fitness value $F(\mathbf{s})$ found versus the iteration number. The hybrid initialization strategy employed by PC-DOA provides a superior starting fitness compared to standard random initialization. While both algorithms exhibit convergence, PC-DOA demonstrates consistent improvement throughout the iterations, indicative of the effectiveness of its bi-phase exploration and exploitation strategy.
\begin{figure}[!t]
	\centering
	\includegraphics[width=0.85\linewidth]{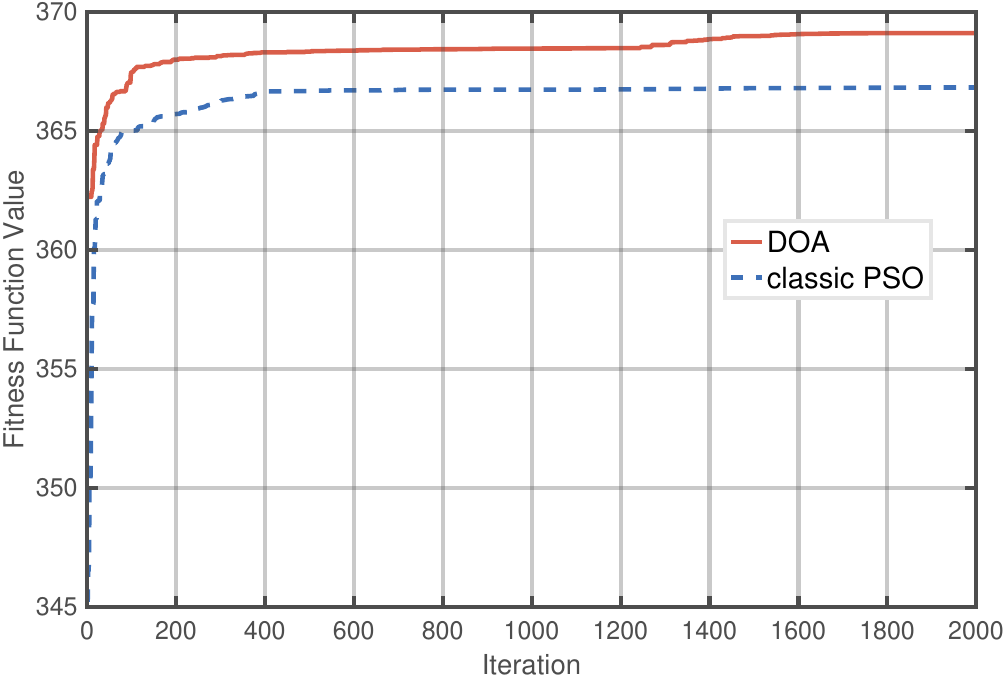}
	\caption{Convergence curves showing the evolution of the best fitness value ($F(\mathbf{s})$) versus iteration count for PC-DOA and PSO.}
	\label{fig:convergence}
\end{figure}
\subsubsection{Optimized Waveform Characteristics}
Finally, we examine the properties of a representative waveform optimized using PC-DOA applied to the MIUB objective. Fig.~\ref{fig:autocorr} displays the magnitude of its autocorrelation function. The waveform achieves effective sidelobe suppression, with a measured Peak Sidelobe Level (PSL) of -22.47 dB relative to the mainlobe peak. Low autocorrelation sidelobes are crucial for resolving closely spaced targets and mitigating self-clutter. Fig.~\ref{fig:ambiguity} shows the corresponding ambiguity function magnitude. The optimized waveform retains a desirable "knife-edge" ridge structure, characteristic of LFM-like signals, suggesting good range resolution and moderate Doppler tolerance, inherited partly from the LFM-inspired initialization and refined during optimization.
\begin{figure}[!t]
	\centering
	\includegraphics[width=0.85\linewidth]{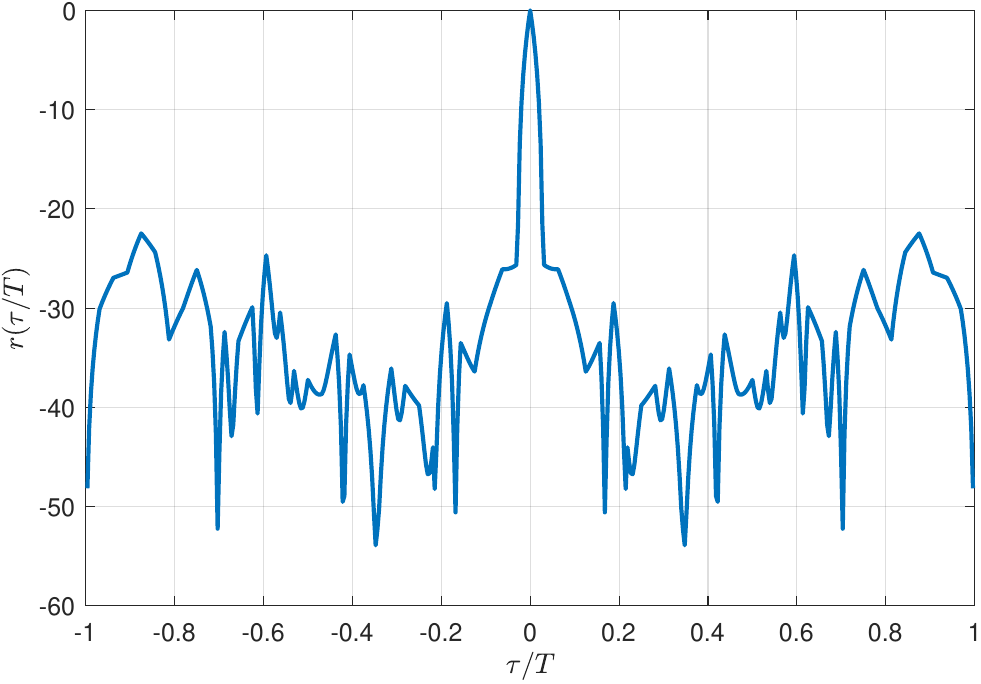}
	\caption{Magnitude of the autocorrelation function for a representative waveform optimized using PC-DOA with the MIUB objective.}
	\label{fig:autocorr}
\end{figure}
\begin{figure}[!t]
	\centering
	\includegraphics[width=0.85\linewidth]{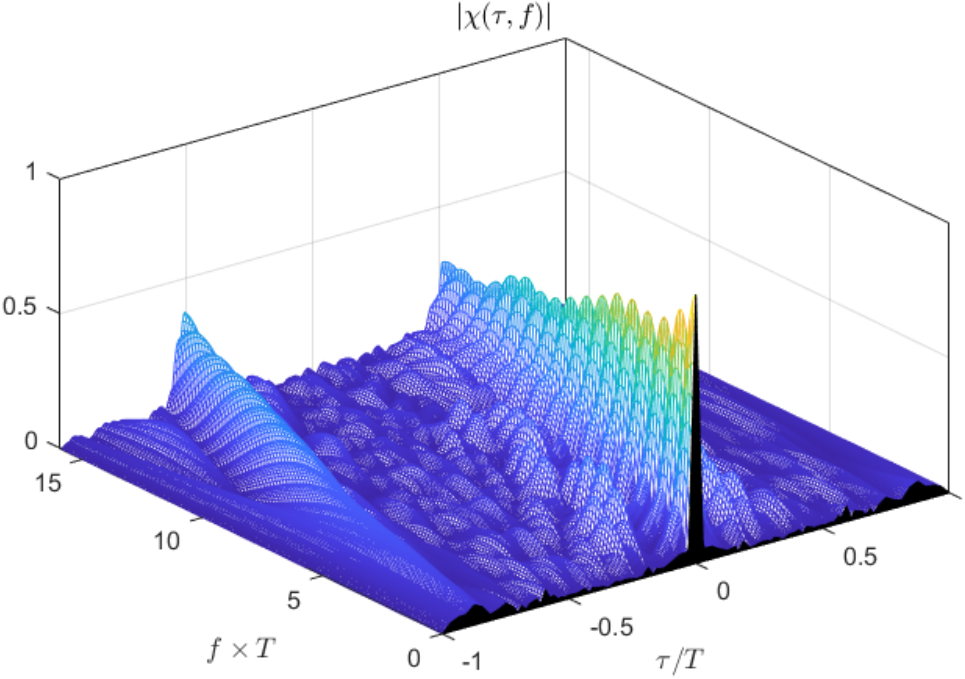}
	\caption{Magnitude of the ambiguity function for a representative waveform optimized using PC-DOA with the MIUB objective.}
	\label{fig:ambiguity}
\end{figure}

\subsection{Discussion}
\label{sec4.3}
The numerical results validate that the proposed approach can effectively balance the inherent detection–estimation trade-offs. Specifically, by optimizing the MIUB objective $F(\mathbf{s})$, the tailored PC-DOA algorithm synthesizes constant modulus waveforms with superior detection performance compared to direct mutual information maximization, though with a slight reduction in estimation accuracy. This behavior is in line with our theoretical predictions.

Notably, the PC-DOA algorithm demonstrates faster convergence and better final solution quality than traditional methods such as PSO. Its hybrid initialization and adaptive exploration–exploitation scheme on the complex circle manifold are particularly effective. Furthermore, the resulting waveforms exhibit excellent autocorrelation and ambiguity function characteristics, making them well-suited for practical radar systems that require both clutter suppression and Doppler tolerance.

\textbf{Limitations and Future Work:} The reliance on Gaussian Mixture Distributions (GMDs) for modeling target and clutter statistics, while justified by prior studies \cite{blacknell2000target, chen_joint_2023, gu2019information} and enabling analytical tractability for the MIUB approximations, represents a modeling assumption. Although GMDs offer considerable flexibility, scenarios involving extreme non-Gaussian interference or highly specific target characteristics might not be perfectly captured. However, the simulation results across various SCRs suggest robustness, implying moderate deviations from the strict GMD assumption may have limited impact. Future work could explore robust formulations or alternative statistical models (e.g., using normalizing flows or variational autoencoders) compatible with information-theoretic objectives. The current simulations also adopted standard simplifications (e.g., LTI channel, specific assumptions on frequency selectivity implicit in the PSDs). Extending the analysis and validation to more complex, time-varying, or spatially diverse (MIMO) channels is another direction for future research. Finally, the computational complexity, while analyzed in Section III.D, scales with the number of GMD components ($L$) and waveform length ($N$), potentially limiting applicability in real-time scenarios with very high-dimensional models; investigating computational speedups or alternative optimization approaches for large-scale problems remains relevant.


\section{Conclusion} \label{sec5}

In this paper, we proposed a unified information-theoretic framework for radar waveform design under constant modulus constraints, simultaneously addressing target detection and parameter estimation. By modeling the target and clutter impulse responses with Gaussian Mixture Distributions (GMDs), we derived tractable approximations for both the Kullback–Leibler divergence and the Mutual Information Upper Bound (MIUB), thereby avoiding heuristic scalarization.

The proposed Phase-Coded Dream Optimization Algorithm (PC-DOA) effectively solves the resulting non-convex optimization problem via a novel hybrid initialization and adaptive search over the complex circle manifold. Extensive numerical evaluations confirm that the MIUB-based design outperforms conventional methods in achieving a desirable trade-off between detection and estimation metrics. Furthermore, the optimized waveforms possess favorable ambiguity function properties, such as enhanced Doppler tolerance and reduced sidelobe levels.

\section*{Acknowledgments}

The authors utilized Deepseek (https://www.deepseek.com) for English language polishing and grammatical refinement of the manuscript. Specifically, Deepseek was employed to enhance the clarity, coherence, and readability of the text in the Introduction, Discussion, and Conclusion sections. All technical content, equations, figures, and conclusions remain solely the work of the authors. The use of AI-generated content was limited to non-substantive linguistic improvements and did not involve the generation or alteration of scientific ideas, results, or methodologies.

\appendices


\section{Proof of Approximation Error Bounds for MI and KL Divergence}
\label{app:error_analysis}

In this appendix, we provide a detailed analysis of the error bounds associated with the Mutual Information (MI) Approximation (Lemma~\ref{lemma2_concise}) and the Kullback-Leibler (KL) Divergence Approximation (Lemma~\ref{lemma3_concise}).

\subsection{Error Bound of MI Approximation}

The exact mutual information is given by the expression 
\begin{equation}
    I(\mathbf{x};\mathbf{y}) = h(\mathbf{y}) - h(\mathbf{y}|\mathbf{x}),
\end{equation}
where $h(\cdot)$ denotes differential entropy. The approximation $\overline{\mathfrak{E}}(\mathbf{s})$ emerges from estimating the differential entropy $h(\mathbf{y})$ via a first-order Taylor expansion of the log-likelihood function centered at the mean $\mathbf{y}_0 = \mathbb{E}[\mathbf{y}]$. 

Let $p(\mathbf{y})$ represent the true probability density function of $\mathbf{y}$ under hypothesis $\mathcal{H}_1$, corresponding to $p_1(\mathbf{y})$ as specified in \eqref{eq:p1_gmd}. The exact differential entropy is then expressed as 
\begin{equation}
    h(\mathbf{y}) = -\mathbb{E}_{\mathbf{y}}[\log p(\mathbf{y})].
\end{equation}

The approximation introduced in \cite{gu2019information} effectively substitutes $h(\mathbf{y})$ with $\tilde{h}(\mathbf{y}) = -\log p(\mathbf{y}_0)$. This simplification assumes $\mathbf{y}_0 = \mathbf{0}$ for zero-mean stochastic processes and leverages properties of Gaussian determinants. The discrepancy in approximating $I(\mathbf{x};\mathbf{y})$ predominantly originates from the error in estimating $h(\mathbf{y})$. Defining the entropy approximation error as 
\[
\Delta h = h(\mathbf{y}) - \tilde{h}(\mathbf{y})
\]
we can quantify the resulting deviation in the mutual information estimate.

\subsubsection{Taylor Expansion Remainder Analysis}

Consider the log-likelihood function defined as  
\begin{equation}
    \phi(\mathbf{y}) = \log p_1(\mathbf{y}) = \log \left[ \sum_{\ell=1}^{L} \gamma_\ell \, \mathcal{CN}(\mathbf{y}; \mathbf{0}, \mathbf{\Sigma}_\ell) \right].
\end{equation}

Performing a second-order Taylor expansion of $\phi(\mathbf{y})$ around the point $\mathbf{y}_0 = \mathbf{0}$ yields the approximation  
\begin{align}
    \phi(\mathbf{y}) \approx \phi(\mathbf{0}) + \nabla \phi(\mathbf{0})^{\mathrm{H}} \mathbf{y} + \frac{1}{2} \mathbf{y}^{\mathrm{H}} \mathbf{H}(\mathbf{0}) \mathbf{y},
\end{align}
where $\mathbf{H}(\mathbf{0})$ denotes the Hessian matrix of $\phi(\mathbf{y})$ evaluated at $\mathbf{y} = \mathbf{0}$.  

Taking the expectation with respect to $\mathbf{y} \sim p_1(\mathbf{y})$, we obtain  
\begin{align}
    \mathbb{E}[\phi(\mathbf{y})] \approx \phi(\mathbf{0}) + \mathbb{E}[\nabla \phi(\mathbf{0})^{\mathrm{H}} \mathbf{y}] + \frac{1}{2} \mathbb{E}[\mathbf{y}^{\mathrm{H}} \mathbf{H}(\mathbf{0}) \mathbf{y}].
\end{align}

Given that the differential entropy satisfies $h(\mathbf{y}) = -\mathbb{E}[\phi(\mathbf{y})]$ and the approximation employs $\tilde{h}(\mathbf{y}) = -\phi(\mathbf{0})$, the dominant contribution to the approximation error arises from the second-order term:  
\begin{align}
    \Delta h &= h(\mathbf{y}) - \tilde{h}(\mathbf{y}) \approx -\frac{1}{2} \mathbb{E}[\mathbf{y}^{\mathrm{H}} \mathbf{H}(\mathbf{0}) \mathbf{y}] \nonumber \\
             &= -\frac{1}{2} \mathrm{tr}\left( \mathbf{H}(\mathbf{0}) \mathbb{E}[\mathbf{y}\mathbf{y}^{\mathrm{H}}] \right) = -\frac{1}{2} \mathrm{tr}\left( \mathbf{H}(\mathbf{0}) \mathbf{C}_{\mathbf{y}} \right),
    \label{eq:apxMIremainder_app}
\end{align}
where $\mathbf{C}_{\mathbf{y}} = \mathbb{E}[\mathbf{y}\mathbf{y}^{\mathrm{H}}] = \sum_{\ell=1}^{L} \gamma_\ell \mathbf{\Sigma}_\ell$ represents the aggregate covariance matrix of $\mathbf{y}$ under hypothesis $\mathcal{H}_1$.

\subsubsection{Error Bound Derivation}

Using the property $|\mathrm{tr}(\mathbf{A}\mathbf{B})| \leq \|\mathbf{A}\|_{\mathrm{F}} \|\mathbf{B}\|_{\mathrm{F}}$ for Hermitian matrices, we can bound the magnitude of the error:
\begin{equation}
  |\Delta h| \leq \frac{1}{2} \|\mathbf{H}(\mathbf{0})\|_{\mathrm{F}} \|\mathbf{C}_{\mathbf{y}}\|_{\mathrm{F}}.
  \label{eq:MIboundF_app}
\end{equation}

If the spectral norm (largest eigenvalue magnitude) of the Hessian is bounded, $\|\mathbf{H}(\mathbf{0})\|_2 \leq L_H$, then using $\|\mathbf{A}\|_{\mathrm{F}} \leq \sqrt{\mathrm{rank}(\mathbf{A})} \|\mathbf{A}\|_2$, we get $\|\mathbf{H}(\mathbf{0})\|_{\mathrm{F}} \leq \sqrt{n'} L_H$, where $n' = N+N_T-1$. This leads to:
\begin{equation}
  |\Delta h| \leq \frac{\sqrt{n'} L_H}{2} \|\mathbf{C}_{\mathbf{y}}\|_{\mathrm{F}}.
\end{equation}

Alternatively, using $|\mathrm{tr}(\mathbf{A}\mathbf{B})| \leq \|\mathbf{A}\|_2 \mathrm{tr}(\mathbf{B})$ if $\mathbf{B}$ is positive semidefinite:
\begin{equation}
  |\Delta h| \leq \frac{L_H}{2} \mathrm{tr}(\mathbf{C}_{\mathbf{y}}).
  \label{eq:MIboundTrace_app}
\end{equation}

The constant $L_H$ depends on the properties of the GMM components (weights and covariances). The error is thus controlled by the overall variance $\mathrm{tr}(\mathbf{C}_{\mathbf{y}})$ and the curvature of the log-likelihood at the origin.

\subsection{Error Bound of KL Divergence Approximation}

The exact Kullback-Leibler (KL) divergence between the distributions $p_1$ and $p_0$ is denoted by $\mathcal{D}_{\mathrm{KL}}(p_1 \parallel p_0)$. The approximation $\overline{\mathfrak{D}}(\mathbf{s})$, presented in \eqref{eq16_concise}, is constructed by associating each component $\ell$ of $p_1$ with its "closest" counterpart $k^\star(\ell)$ in $p_0$ as determined by the cost function $J(k, \ell)$ defined in \eqref{eq:J_k_ell_concise}. The resulting approximation error is quantified as 
\begin{equation}
    \Delta \mathcal{D} = \mathcal{D}_{\mathrm{KL}}(p_1 \parallel p_0) - \overline{\mathfrak{D}}(\mathbf{s}).
\end{equation}

\subsubsection{Matching and Cross-Term Errors}

The true KL divergence can be written as:
\begin{align}
  &\mathcal{D}_{\mathrm{KL}}(p_1 \parallel p_0) = \int p_1(\mathbf{y}) \log \frac{p_1(\mathbf{y})}{p_0(\mathbf{y})} d\mathbf{y} \nonumber \\
  &= \sum_{\ell=1}^{L} \gamma_\ell \int \mathcal{CN}(\mathbf{y};\mathbf{0},\mathbf{\Sigma}_\ell) \log \frac{\sum_{m=1}^{L} \gamma_m \mathcal{CN}(\mathbf{y};\mathbf{0},\mathbf{\Sigma}_m)}{\sum_{k=1}^{K} \alpha_k \mathcal{CN}(\mathbf{y};\mathbf{0},\mathbf{R}_k)} d\mathbf{y}.
\end{align}

The approximation essentially replaces the denominator $p_0(\mathbf{y})$ inside the logarithm for the $\ell$-th term's integral with $\alpha_{k^\star(\ell)} \mathcal{CN}(\mathbf{y};\mathbf{0},\mathbf{R}_{k^\star(\ell)})$ and simplifies the numerator. The error arises from two main sources:
\begin{enumerate}
    \item \textbf{Matching Error:} The approximation assumes that for each component $\ell$ of $p_1$, the dominant contribution to the KL divergence comes from its interaction with a single component $k^\star(\ell)$ of $p_0$. This ignores the influence of other components $k \neq k^\star(\ell)$.
    \item \textbf{Cross-Term Error:} The structure of the approximation effectively treats the KL divergence as a sum of divergences between individual matched components, neglecting the complex interplay arising from the mixture nature of both $p_1$ and $p_0$.
\end{enumerate}

Let $p_1^{(\ell)}(\mathbf{y}) = \mathcal{CN}(\mathbf{y};\mathbf{0},\mathbf{\Sigma}_\ell)$ and $p_0^{(k)}(\mathbf{y}) = \mathcal{CN}(\mathbf{y};\mathbf{0},\mathbf{R}_k)$. The error can be expressed, after some manipulation involving Jensen's inequality on the logarithm of sums, as related to terms like:
\begin{align}
  \sum_{\ell=1}^{L} \gamma_\ell \log \frac{\sum_{k=1}^{K} \alpha_k \exp(-\mathcal{D}_{\mathrm{KL}}(p_1^{(\ell)} \parallel p_0^{(k)}))}{\alpha_{k^\star(\ell)}}.
\end{align}

This term captures the deviation due to considering only the $k^\star(\ell)$ component versus the weighted mixture in the denominator.

\subsubsection{Error Bound using Pinsker's Inequality (Matching Error)}
Pinsker's inequality relates KL divergence to the total variation distance 
\[
\|\cdot\|_{\mathrm{TV}}:\|\pi_1 - \pi_2\|_{\mathrm{TV}}^2 \leq \frac{1}{2} \mathcal{D}_{\mathrm{KL}}(\pi_1 \parallel \pi_2)
\]

While not directly bounding $\Delta \mathcal{D}$, it highlights that if a component $p_1^{(\ell)}$ is very "far" (in KL sense) from all $p_0^{(k)}$ except $p_0^{(k^\star(\ell))}$, the matching approximation is more accurate. If 
\[
\min_{k \neq k^\star(\ell)} \mathcal{D}_{\mathrm{KL}}(p_1^{(\ell)} \parallel p_0^{(k)}) \geq \delta > 0
\]
the contribution of mismatched terms is reduced. The error related to the matching choice itself can be bounded if the gap between $J(k^\star(\ell), \ell)$ and $J(k, \ell)$ for $k \neq k^\star(\ell)$ is significant.

\subsubsection{Error Bound based on Component Overlap (Cross-Term Error)}

The error is larger when components of $p_1$ significantly overlap with multiple components of $p_0$. If the components are well-separated, the approximation is better. A bound can be derived based on the properties of the component covariances. If, for instance, there exists $\beta > 0$ such that $\mathbf{\Sigma}_\ell \preceq \beta \mathbf{R}_k$ (in the Loewner order) for all $\ell, k$, then the individual KL divergences $\mathcal{D}_{\mathrm{KL}}(p_1^{(\ell)} \parallel p_0^{(k)})$ are bounded. 

The overall error $\Delta \mathcal{D}$ can be shown to be bounded by a function that decreases as the separation between components increases and increases with the overlap factor $\beta$ and the number of components $L, K$. A rough bound structure is:
\begin{equation}
  |\Delta \mathcal{D}| \leq \mathcal{O}\left( \sum_{\ell=1}^L \gamma_\ell \sum_{k \neq k^\star(\ell)} \alpha_k \exp(-\mathcal{D}_{\mathrm{KL}}(p_1^{(\ell)} \parallel p_0^{(k)})) \right).
\end{equation}
This implies that the error is small if the KL divergences to non-matched components are large (i.e., components are well-separated).

In conclusion, the MI approximation error is primarily governed by the local curvature and overall spread of the GMM, while the KL approximation error is mainly influenced by the distinctiveness and separation of the mixture components being compared. Both approximations are reasonable under conditions where these factors are controlled (e.g., low variance for MI, well-separated components for KL).


\section{Proof of Proposition~\ref{prop1}}
\label{sec:lipschitz_proof}

We prove that the objective function
\begin{equation}
F(\mathbf{s}) = \overline{\mathfrak{D}}(\mathbf{s}) + \overline{\mathfrak{E}}(\mathbf{s})
\end{equation}
is Lipschitz continuous over the feasible set 
\begin{equation}
\mathcal{M} = \{\mathbf{s} \in \mathbb{C}^N : |s_n| = c, \ 
\forall n=1,...,N\}.
\end{equation}

Assume the component covariance matrices, $\{\mathbf{R}_k\}$ and $\{\mathbf{Q}_m\}$, are positive definite with all eigenvalues in $[\lambda_{\min}, \lambda_{\max}]$. This also guarantees that the composite covariance 
\(
\mathbf{\Sigma}_\ell(\mathbf{s}) = \mathbf{S}\mathbf{Q}_m\mathbf{S}^H + \mathbf{R}_k
\)
remains positive definite, with its eigenvalues in $[\lambda'_{\min}, \lambda'_{\max}]$ for all feasible $\mathbf{s}$. Here, $\mathbf{S}(\mathbf{s})$ denotes the matrix form generated from $\mathbf{s}$ (e.g., via a Toeplitz operation), and $N_T$ is the transmit dimension associated with $\mathbf{S}$.

\subsection{Lipschitz Continuity of the Composite Covariance Matrix}

Consider any two feasible vectors $\mathbf{s}_1,\mathbf{s}_2 \in \mathcal{M}$. Their respective matrix representations are $\mathbf{S}_1 = \mathbf{S}(\mathbf{s}_1)$ and $\mathbf{S}_2 = \mathbf{S}(\mathbf{s}_2)$. By properties of the mapping from $\mathbf{s}$ to $\mathbf{S}$ and the constant modulus,
\begin{equation}
\|\mathbf{S}_1 - \mathbf{S}_2\|_{\mathrm{F}} = \sqrt{N_T}\|\mathbf{s}_1 - \mathbf{s}_2\|_2.
\label{eq:StoS}
\end{equation}

Now consider the difference of two composite covariance matrices:
\begin{align}
&\mathbf{\Sigma}_\ell(\mathbf{s}_1) - \mathbf{\Sigma}_\ell(\mathbf{s}_2) \nonumber \\
&\quad = \mathbf{S}_1\mathbf{Q}_m\mathbf{S}_1^H - \mathbf{S}_2\mathbf{Q}_m\mathbf{S}_2^H \nonumber \\
&\quad = (\mathbf{S}_1 - \mathbf{S}_2)\mathbf{Q}_m\mathbf{S}_1^H + 
        \mathbf{S}_2\mathbf{Q}_m(\mathbf{S}_1^H - \mathbf{S}_2^H).
\end{align}

Taking the Frobenius norm and using triangle and submultiplicative inequalities, we have
\begin{align}
&\|\mathbf{\Sigma}_\ell(\mathbf{s}_1) - \mathbf{\Sigma}_\ell(\mathbf{s}_2)\|_{\mathrm{F}} \nonumber \\
&\leq \|(\mathbf{S}_1 - \mathbf{S}_2)\mathbf{Q}_m\mathbf{S}_1^H\|_{\mathrm{F}}
   + \|\mathbf{S}_2\mathbf{Q}_m(\mathbf{S}_1^H - \mathbf{S}_2^H)\|_{\mathrm{F}} \nonumber \\
&\leq \|\mathbf{S}_1 - \mathbf{S}_2\|_{\mathrm{F}} \cdot \|\mathbf{Q}_m\|_{\mathrm{F}} \cdot \|\mathbf{S}_1^H\|_{\mathrm{F}} 
     \nonumber \\
& \quad + \|\mathbf{S}_2\|_{\mathrm{F}} \cdot \|\mathbf{Q}_m\|_{\mathrm{F}} \cdot \|\mathbf{S}_1 - \mathbf{S}_2\|_{\mathrm{F}}.
\label{eq:covbound}
\end{align}

Because all $\mathbf{s} \in \mathcal{M}$ have constant modulus, there exists $E_s > 0$, such that $\|\mathbf{S}_i\|_{\mathrm{F}} = \sqrt{N_T E_s}$ for $i=1,2$. Also, $\|\mathbf{Q}_m\|_{\mathrm{F}} \leq \lambda_{\max} n_T$ by spectral norm.
Combining, we obtain:
\begin{align}
&\|\mathbf{\Sigma}_\ell(\mathbf{s}_1) - \mathbf{\Sigma}_\ell(\mathbf{s}_2)\|_{\mathrm{F}} \nonumber \\
&\leq 2\sqrt{N_T E_s}\lambda_{\max} \sqrt{N_T} \|\mathbf{s}_1 - \mathbf{s}_2\|_2 \nonumber \\
&= 2 N_T \sqrt{E_s} \lambda_{\max} \|\mathbf{s}_1-\mathbf{s}_2\|_2.
\end{align}
Define $L_{\Sigma} = 2N_T\sqrt{E_s}\lambda_{\max}$, thus, 
\begin{equation}
\|\mathbf{\Sigma}_\ell(\mathbf{s}_1) - \mathbf{\Sigma}_\ell(\mathbf{s}_2)\|_{\mathrm{F}} 
\leq L_{\Sigma}\|\mathbf{s}_1 - \mathbf{s}_2\|_2.
\end{equation}

\subsection{Lipschitz Continuity of the Objective Function Components}

\subsubsection{Mutual Information Term $\overline{\mathfrak{E}}(\mathbf{s})$}

Recall 
\(
\overline{\mathfrak{E}}(\mathbf{s}) = \log C_R - \log G(\mathbf{s})
\)
with 
\(
G(\mathbf{s}) = \sum_{\ell}\gamma_\ell\det(\mathbf{\Sigma}_\ell(\mathbf{s}))^{-1}.
\)
Consider two feasible vectors $\mathbf{s}_1,\mathbf{s}_2$:
\begin{align}
|G(\mathbf{s}_1) &- G(\mathbf{s}_2)| \nonumber \\
&\leq \sum_{\ell} \gamma_\ell 
  \left| \det(\mathbf{\Sigma}_\ell(\mathbf{s}_1))^{-1} 
        - \det(\mathbf{\Sigma}_\ell(\mathbf{s}_2))^{-1} \right| \nonumber \\
&\leq \sum_{\ell} \gamma_\ell \frac{
\left| \det(\mathbf{\Sigma}_\ell(\mathbf{s}_2)) - 
       \det(\mathbf{\Sigma}_\ell(\mathbf{s}_1)) \right|
}
{\min\{ \det(\mathbf{\Sigma}_\ell(\mathbf{s}_1)), \det(\mathbf{\Sigma}_\ell(\mathbf{s}_2)) \}^2}.
\end{align}

Matrix analysis (see e.g.~\cite{horn2012matrix}) shows 
\[
|\det(\mathbf{A}) - \det(\mathbf{B})| \leq D \|\mathbf{A}-\mathbf{B}\|_\mathrm{F}
\]
if both $\mathbf{A},\mathbf{B}$ are Hermitian with all eigenvalues in 
$[\lambda'_{\min},\lambda'_{\max}]$, where $D$ depends only on $n'$, $\lambda'_{\min}$, $\lambda'_{\max}$. Thus, for some constant $L_h$,
\begin{equation}
|G(\mathbf{s}_1) - G(\mathbf{s}_2)| \leq L_h L_\Sigma \|\mathbf{s}_1-\mathbf{s}_2\|_2.
\end{equation}

Since $-\log(\cdot)$ is Lipschitz continuous for arguments bounded below ($G(\mathbf{s}) \geq \epsilon > 0$ in $\mathcal{M}$), $\overline{\mathfrak{E}}(\mathbf{s})$ is Lipschitz with constant
\begin{equation}
L_{\overline{E}} = \frac{L_h L_\Sigma}{\epsilon}.
\end{equation}

\subsubsection{KL Divergence Term $\overline{\mathfrak{D}}(\mathbf{s})$}

Recall
\[
J(k,\ell,\mathbf{s}) = \log \frac{\gamma_\ell}{\alpha_k} 
+ \mathrm{tr}\left(\mathbf{R}_k^{-1} \mathbf{\Sigma}_\ell(\mathbf{s})\right)
- \log \det (\mathbf{R}_k^{-1}\mathbf{\Sigma}_\ell(\mathbf{s}))
\]

For the trace part, using norm inequalities,
\begin{align}
|\mathrm{tr}(\mathbf{R}_k^{-1} \Delta \mathbf{\Sigma}_\ell)| 
&\leq \|\mathbf{R}_k^{-1}\|_2 \,\mathrm{tr}(|\Delta \mathbf{\Sigma}_\ell|) \nonumber \\
&\leq n' \lambda_{\min}^{-1} \|\Delta \mathbf{\Sigma}_\ell\|_{\mathrm{F}} \nonumber \\
&\leq n' \lambda_{\min}^{-1} L_\Sigma \|\mathbf{s}_1 - \mathbf{s}_2\|_2.
\end{align}
The log-determinant term is also Lipschitz (see e.g.~\cite{horn2012matrix}):
\begin{align}
&|\log \det (\mathbf{R}_k^{-1}\mathbf{\Sigma}_\ell(\mathbf{s}_1)) 
    - \log \det (\mathbf{R}_k^{-1}\mathbf{\Sigma}_\ell(\mathbf{s}_2))| \nonumber \\
&\qquad \leq n' (\lambda'_{\min})^{-1} L_\Sigma \|\mathbf{s}_1 - \mathbf{s}_2\|_2.
\end{align}

Summing both contributions, $J(k,\ell,\mathbf{s})$ is Lipschitz with
\begin{equation}
L_{J,k,\ell} = \left( n' (\lambda_{\min})^{-1} + n' (\lambda'_{\min})^{-1} \right) 
               L_\Sigma.
\end{equation}

Since $\overline{\mathfrak{D}}(\mathbf{s})$ is constructed as a weighted sum, its Lipschitz constant is
\begin{equation}
L_{\overline{D}} = \max_\ell \sum_k \alpha_k \gamma_\ell L_{J,k,\ell}.
\end{equation}

In short, since both $\overline{\mathfrak{E}}(\mathbf{s})$ and $\overline{\mathfrak{D}}(\mathbf{s})$ are Lipschitz continuous on $\mathcal{M}$, their sum $F(\mathbf{s})$ is also Lipschitz on $\mathcal{M}$, i.e.,
\begin{equation}
\left| F(\mathbf{s}_1) - F(\mathbf{s}_2) \right| 
\leq L_{\mathrm{F}} \|\mathbf{s}_1 - \mathbf{s}_2\|_2, \quad \forall \mathbf{s}_1, \mathbf{s}_2 \in \mathcal{M},
\end{equation}
where $L_{\mathrm{F}} = L_{\overline{E}} + L_{\overline{D}}$.


\bibliographystyle{IEEEtran}
\small
\bibliography{refs}

%
%
%
%

\vfill

\end{document}